\documentclass[twocolumn,english,aps,reprint, superscriptaddress,showpacs,showkeys]{revtex4-2}
\usepackage{amsmath,amssymb,bbm,mathrsfs,bm,braket,color,graphicx,comment,todonotes, tikz, booktabs, quantikz}
\usepackage[colorlinks,citecolor=blue,urlcolor=blue,linkcolor = blue]{hyperref}
\usepackage{soul}
\usepackage[mathscr]{euscript}
\usepackage{hyperref}
\usepackage{tikz}
\usepackage{xcolor}
\usepackage{multirow}

\usepackage{enumerate}

\usepackage{amsthm}
\newtheorem{thm}{Theorem}
\newtheorem{lem}{Lemma}
\theoremstyle{definition}
\newtheorem{defi}{Definition}

\newcommand{\neigh}[1]{\mathcal{N}(#1)}
\newcommand{\nodd}[1]{\mathsf{Odd}(#1)}
\newcommand{\cz}{\mathrm{CZ}}
\newcommand{\ident}{\mathbb{I}}
\newcommand{\avg}[1]{\langle #1 \rangle}

\begin{document}
\title{YZ-plane measurement-based quantum computation: Universality and Parity Architecture implementation}

\author{Jaroslav Kysela}
\affiliation{Parity Quantum Computing GmbH, A-6020 Innsbruck, Austria}

\author{Katharina Ludwig}
\affiliation{Parity Quantum Computing Germany GmbH, D-20095 Hamburg, Germany}

\author{Nitica Sakharwade}
\affiliation{Parity Quantum Computing GmbH, A-6020 Innsbruck, Austria}

\author{Anette Messinger}
\affiliation{Parity Quantum Computing GmbH, A-6020 Innsbruck, Austria}

\author{Wolfgang Lechner}
\affiliation{Parity Quantum Computing GmbH, A-6020 Innsbruck, Austria}
\affiliation{Parity Quantum Computing Germany GmbH, D-20095 Hamburg, Germany}
\affiliation{Institute for Theoretical Physics, University of Innsbruck, A-6020 Innsbruck, Austria}

\date{\today}

\begin{abstract}
We define the class of register-logic graphs and prove that any uniformly deterministic measurement-based quantum computation (MBQC) where the inputs coincide with the outputs must be driven on such graphs by measurements in the $YZ$ plane of the Bloch sphere. This observation is revisited in the context that goes beyond uniform determinism, where we present a universal $YZ$-plane-only measurement pattern and establish a connection between $YZ$-plane-only and $XZ$-plane-only patterns. These results conclude the line of research on universal patterns with measurements restricted to one of the principal planes of the Bloch sphere. We further demonstrate, within the framework of the Parity Architecture, that $YZ$-plane patterns with the register-logic graph can be embedded into another graph with purely local interactions, and we extend this case to the scenario of universal quantum computation. 
\end{abstract}

\maketitle

\section{Introduction}

Measurement-based quantum computation (MBQC), as put forward in the works of Raussendorf, Briegel, and collaborators \cite{raussendorf_one-way_2001,raussendorf_computational_2002}, is a paradigm of quantum computation where the execution of a given algorithm is driven by single-qubit measurements performed in a prescribed order on a large entangled quantum state, referred to as a graph state. Since its inception, much effort has been put into exploring its connections to the established gate-based picture of quantum computation \cite{simmons_relating_2021} as well as in studying its intrinsic properties such as the conditions under which a given measurement-based computation leads to deterministic results \cite{browne_generalized_2007,mhalla_shadow_2025}. The measurements are not only the driving force of the measurement-based computation, but also play a prominent role in various techniques of fault-tolerant quantum computing that lie outside the scope of the MBQC formalism \cite{bombin_quantum_2009}. One such technique based on measurements was shown to lead to significant savings in resources when applied to the Parity Architecture \cite{messinger_constant_2023}. In this architecture, the parity of qubits is addressed by local operations to design hardware-friendly quantum algorithms by eliminating the need for long-range interactions \cite{lechner_quantum_2015,dreier_connectivity-aware_2025}. Its resource-efficient modification of Ref.~\cite{messinger_constant_2023} was later proven to be equivalent to a class of measurement-based patterns \cite{smith_parity_2024} where all measurements are restricted to the $YZ$-plane of the Bloch sphere. These patterns were further studied from the purely MBQC perspective and their minimum possible implementation was identified at the cost of introducing long-range interactions, rendering the resulting algorithm less amenable to hardware platforms \cite{smith_minimally_2025}. 

Inherent in the measurement-based implementation of any algorithm is the indeterminacy of quantum measurements. For the implementation to give deterministic results, multiple constraints must be imposed on the underlying graph state as well as the form of measurements \cite{browne_generalized_2007,mhalla_shadow_2025}. The latter is restricted to lie in the principal planes of the Bloch sphere, while the special cases of the Pauli bases that lie in the intersection of two planes are also allowed. In the original MBQC proposal \cite{raussendorf_one-way_2001,raussendorf_computational_2002}, many qubits are measured in the eigenbasis of the Pauli $Z$ operator and the rest of qubits in bases that lie in the $XY$-plane of the Bloch sphere. It was shown later on that given a graph state whose underlying graph is a regular rectangular grid, a so-called cluster state, the $XY$-plane measurements alone are sufficient for achieving universal measurement-based quantum computation \cite{mantri_universality_2017}. In a similar vein, the triangular-grid graph state was shown to be sufficient for real-valued universal MBQC where all measurements are restricted to the $XZ$-plane. The last possibility, namely when all qubits are measured in the $YZ$-plane, was more closely considered only recently \cite{smith_parity_2024}. One of the consequences of that work was that in such a case no universal computation is possible since the underlying graph has to be of a very restricted form.

Here, we further expand on the discussion of $YZ$-plane-only measurement-based computation and its connection to parity computing. After introducing the necessary terminology in Sec.~\ref{sec:terminology} we improve in Sec.~\ref{sec:yz_mbqc_with_gflow} the results of Ref.~\cite{smith_parity_2024}. We prove that measuring qubits in the $YZ$-plane is a direct consequence of the natural requirement that the quantum computation is uniformly deterministic, provided that the input qubits are assumed to coincide with the outputs. Then in Sec.~\ref{sec:parity_architecture} we present a more systematic exposition of the Parity Architecture in the language of measurement-based quantum computing \cite{messinger_constant_2023,smith_parity_2024} with the focus on eliminating long-range interactions. This allows embedding the $YZ$-plane patterns into graphs with a rectangular grid. Finally, in Sec.~\ref{sec:univ_yz} we demonstrate that universal quantum computation is indeed possible for MBQC where all measurements are restricted to the $YZ$-plane. For that, a more relaxed notion of determinism than the one in Ref.~\cite{smith_parity_2024} is used and two different types of universality are considered. For strict universality the universal pattern is presented explicitly, while for computational universality the issue of the $YZ$-plane MBQC is reduced to the resolved issue for the $XZ$-plane MBQC \cite{mhalla_graph_2012}. The work concludes in Sec.~\ref{sec:conclusion}. 

\section{Preliminaries}
\label{sec:terminology}

A given MBQC algorithm is determined by the initial graph state, defined by the underlying graph $G$ with specified sets $I$ and $O$ of input and output qubits together with a sequence of single-qubit measurements, where each measurement is specified by the principal plane ($XY$, $XZ$, or $YZ$) or the principal axis ($X$, $Y$, or $Z$) of the Bloch sphere in which the measurement basis lies. This specification is captured by the function $\lambda$ that to each non-output qubit assigns the corresponding measurement plane or axis. For convenience, by $\lambda \equiv YZ$ we denote the constant function that to each non-output assigns the $YZ$-plane. The triple $(G, I, O)$ is usually referred to as an \emph{open graph}, and the quadruple $(G, I, O, \lambda)$ as a \emph{labeled open graph}. These notions abstract away from the particular angle $\theta$ that a general measurement basis forms with one of the principal axes. As in Ref.~\cite{backens_there_2021} we always assume that each connected component of the open graph contains at least one input or output. Any connected component without inputs and outputs would contribute only an unimportant scalar factor. 

For measurements the following notation is used: the measurement in a basis that lies in the $YZ$-plane and forms the angle of $\theta$ with the $Z$ axis is denoted by $M_{YZ}(\theta)$; the rotation through $\theta$ about the $X$ axis is denoted by $R_X(\theta)$; and the measurement in the eigenbasis of the Pauli $Y$ operator, or $Y$ measurement for short, is denoted by $M_Y$. For other bases the notation extends in a natural way. The entirety of an MBQC algorithm that includes not only the aforementioned objects but also the correction operators is referred to as a \emph{measurement pattern}. Further details on the basics of MBQC together with relevant theoretical results are summarized in Sec.~\ref{sec:prelims}.

\section{YZ-plane MBQC with gflow}
\label{sec:yz_mbqc_with_gflow}

The deterministic execution of a given algorithm when implemented by a measurement pattern is possible as long as the latter supports some kind of correction strategy that handles undesired measurement outcomes. A prominent type of such strategies ensures so-called \emph{uniform determinism}, where the computation remains deterministic irrespective of the rotation angle of the measurement basis of any qubit. This type of strategies is captured by the mathematical structure known as \emph{gflow} \cite{browne_generalized_2007}. It follows directly from the definition of gflow that any input qubit that is not simultaneously an output must be measured in the $XY$-plane \cite{hamrit_reversibility_2015, backens_there_2021}. As our focus is on the patterns with $YZ$-plane measurements, one obtains the following corollary:
\begin{lem}
    \label{thm:inputs_gflow}
    Any labeled open graph $(G, I, O, \lambda)$ with gflow and $\lambda \equiv YZ$ must satisfy $I \subseteq O$.
\end{lem}
In contrast to the usually considered measurement patterns with distinct inputs and outputs, this statement shows that the computation driven by $YZ$-plane measurements corresponds to the degenerate case where all inputs are simultaneously outputs and remain unmeasured. When $I = O$, the output qubits can therefore be seen as a quantum register that carries the processed information, while the measurements of the non-output qubits represent the quantum logic applied to this information. Moreover, for such a computation to be uniformly deterministic, there must be no edges between the non-output qubits \cite{smith_parity_2024}. These observations motivate the following definition that applies also outside the scope of $YZ$-plane MBQC:
\begin{defi}
    An open graph $(G, I, O)$ is of the \emph{register-logic form} (or \emph{RL form} for short) if the subgraph of $G$ induced by the non-output qubits $\Bar{O}$ has no edges. If also the subgraph of $G$ induced by $O$ has no edges, then the open graph is of the \emph{bipartite-register-logic form} (or \emph{bRL form} for short). Such open graphs are also called \emph{RL graphs} and \emph{bRL graphs}, respectively.
\end{defi}
In Fig.~\ref{fig:rl_state}(a) an example of an open graph of RL form is shown. Note that a bRL graph is bipartite with its partition given by $\Bar{O}$ and $O$, cf.~Fig.~\ref{fig:rl_state}(b). As a side note, the \emph{local complementation} \cite{hein_entanglement_2006} applied to any non-output qubit of an RL graph results again in an RL graph. When instead applied to an output qubit, the result is an RL graph if and only if the qubit has at most one neighbor among the non-output qubits. The local complementation is a graph operation that turns out to be equivalent to a specific local Clifford gate acting on the corresponding graph state. Notably, the $Y$ measurement of a given qubit can be described like its local complementation together with extra $S$ and $Z$ gates acting on its neighbors \cite{hein_entanglement_2006}, cf. also Sec.~\ref{sec:comp_univ}.

\begin{figure}
    \centering
    \includegraphics[width=\linewidth]{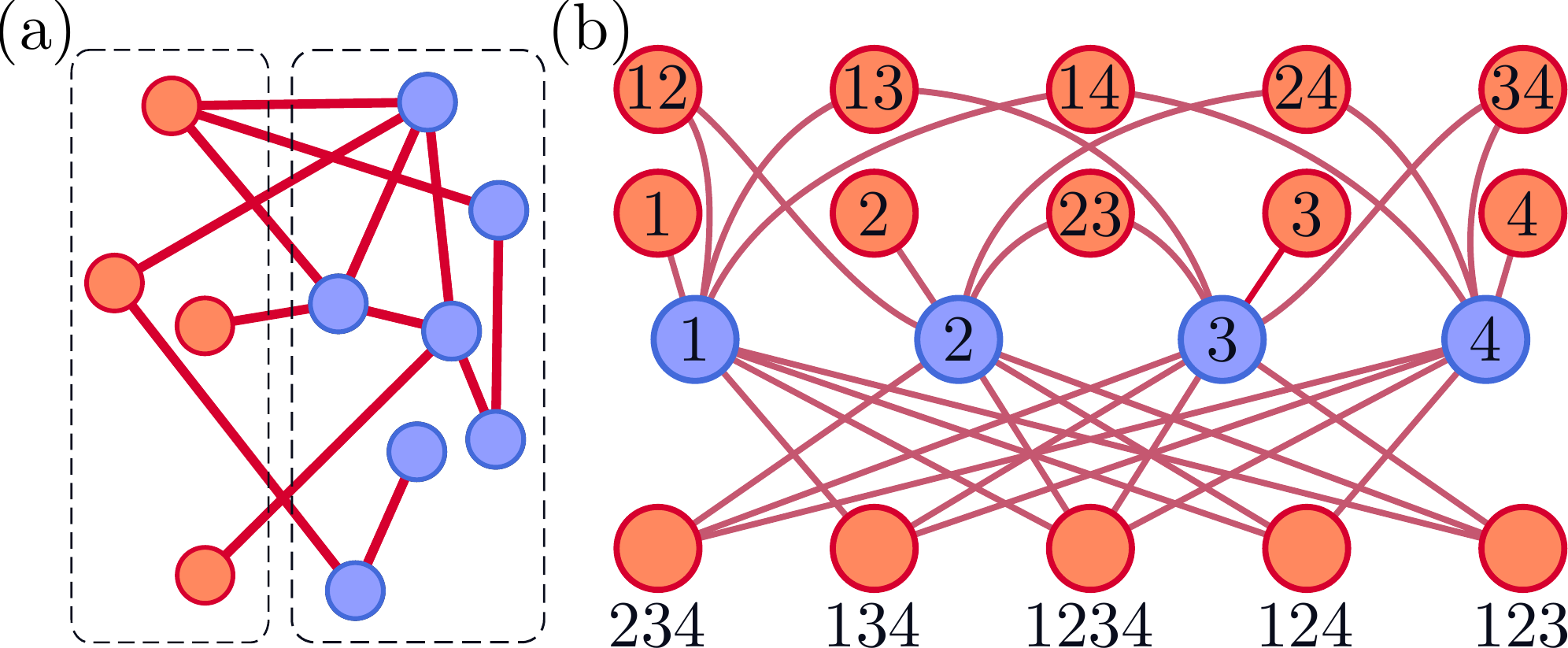}
    \caption{Open graph $(G, I, O)$ of the register-logic form. %
    (a) General RL graph, where $O$ is formed by the nodes in the right box. If there were no edges between nodes in the right box, the open graph would be of the bipartite-register-logic form. %
    (b) bRL graph that allows to implement an arbitrary unitary on four qubits that is diagonal in the $Z$ basis. All orange (dark gray) nodes are measured in the $YZ$-plane and correspond to parity qubits, where the numbers refer to the given parity for each qubit, see Sec.~\ref{sec:parity_architecture}.}
    \label{fig:rl_state}
\end{figure}

As suggested by the above arguments, four important features of measurement patterns with open graph $(G, I, O)$ can be identified that are relevant in the context of $YZ$-plane measurements:
\begin{description}
    \item[$\mathbf{I \subseteq O}$] The sets of inputs and outputs satisfy $I \subseteq O$.
    \item[$\mathbf{\pmb{\lambda} \equiv YZ}$] All non-outputs are measured in the $YZ$-plane.
    \item[RL] The open graph is of the register-logic form.
    \item[gflow] The labeled open graph $(G, I, O, \lambda)$ has gflow.
\end{description}
It is easy to convince oneself that any of the four features alone does not imply any of the others. One can consider a pair of features and study if they imply any of the other two. It turns out that there are exactly two such cases, captured by Lemma~\ref{thm:inputs_gflow} and the following theorem:
\begin{thm}
    \label{thm:io_g_to_rl_yz}
    Let $(G, I, O, \lambda)$ be a labeled open graph with $I = O$ and with gflow. Then $\lambda \equiv YZ$ and $(G, I, O)$ is of the register-logic form.
\end{thm}
This theorem shows that a measurement pattern that has a priori no restrictions on the measurement planes is forced by the mere requirement of the uniform determinism to have measurements only in the $YZ$ plane.

\begin{figure}
    \centering
    \includegraphics[width=0.5\linewidth]{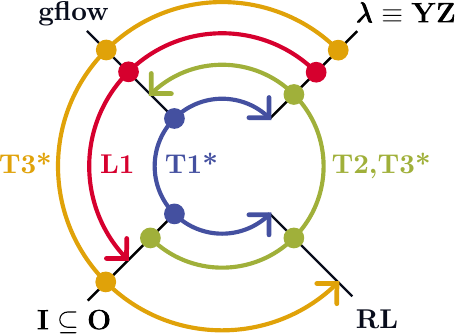}
    \caption{Diagram explaining the mutual relations between Theorems~\ref{thm:io_g_to_rl_yz}, \ref{thm:rl_yz_impl}, and \ref{thm:mod_isaac_thm_1}, denoted in the figure by T1, T2, and T3, respectively, and Lemma~\ref{thm:inputs_gflow} denoted by L1. The asterisk marks that the theorem is valid only for $I = O$. Note that T3 contains two implications, either of which corresponds to one arrow.}
    \label{fig:thm_diag}
\end{figure}

One can further study whether any triple of the features implies the fourth one. It turns out that all four such implications hold true, except that some of them hold only when one restricts condition $I \subseteq O$ to the equality $I = O$. The cases are covered by Lemma~\ref{thm:inputs_gflow}, Theorem~\ref{thm:io_g_to_rl_yz}, and the theorem below:
\begin{thm}
    \label{thm:rl_yz_impl}
    Let $\Gamma = (G, I, O, \lambda)$ be a labeled open graph with $I \subseteq O$ and $\lambda \equiv YZ$, where $(G, I, O)$ is of the register-logic form. Then $\Gamma$ has gflow.
\end{thm}
The proofs of the theorems together with the full discussion are presented in Sec.~\ref{sec:cond_discussion}. For the reader's convenience, in Fig.~\ref{fig:thm_diag} the mutual relations between these claims are exposed. One can compare our findings with those of Ref.~\cite{smith_parity_2024}, where the first systematic discussion of $YZ$-only measurement patterns is made. One of the main results of that work is Theorem~2 that is restated in its supplementary as Theorem~1 and proven there. In Sec.~\ref{sec:cond_discussion}, we present two counterexamples to this theorem and propose the following modification: 
\begin{thm}[Modification of Theorem~2 in \cite{smith_parity_2024}]
    \label{thm:mod_isaac_thm_1}
    Let $\Gamma = (G, I, O, \lambda)$ be a labeled open graph with $I = O$ and $\lambda \equiv YZ$. Then $\Gamma$ has gflow if and only if $(G, I, O)$ is of the register-logic form.
\end{thm}
This theorem is a direct consequence of Theorems~\ref{thm:io_g_to_rl_yz} and \ref{thm:rl_yz_impl}. Compared to the original, two changes are made. First, instead of only requiring the equality of the number of inputs and outputs, here we require the equality of the sets themselves. Second, we replace the condition on $G$ being bipartite with respect to $I$ with being of the register-logic form. The detailed discussion can be found in Sec.~\ref{sec:cond_discussion}.

Note further that as far as the correction strategy due to gflow is concerned, a labeled open graph of the RL form with $I = O$ has the same gflow as the corresponding bRL version, where all edges between nodes in $O$ are removed. The underlying graph $G$ of the original RL graph can be understood as the union of $G_1 = (V, E_1)$ and $G_2 = (O, E_2)$ with $E_1 \cap E_2 = \emptyset$, where $V$ is the set of nodes of $G$, $E_1$ are the edges between $O$ and $\Bar{O}$, and $G_2$ is the subgraph of $G$ induced by $O$. Instead of the graph $G$ with the initial state $\ket{\psi}$ of the input qubits one can consider graph $G_1$ with $\ket{\psi'} \coloneqq \prod_{e \in E_2} \cz_e \ket{\psi}$ as the initial state. This bRL version of the original RL graph is bipartite and it is this restricted class of labeled open graphs that is considered in Ref.~\cite{smith_parity_2024}.

A consequence of the proof of Theorem~\ref{thm:io_g_to_rl_yz} is that all non-output qubits of any labeled open graph with gflow and $I = O$ can be measured simultaneously in a single time step. One can come to the same conclusion by noting that all the gates implemented by measuring non-output qubits in the $YZ$-plane are the controlled phase gates $\exp(i \alpha Z \otimes \ldots \otimes Z)$ for some real $\alpha$, see Sec.~\ref{sec:parity_architecture}, which are all diagonal in the $Z$ basis. All such gates commute and so can be performed simultaneously. As a side note, single-qubit $R_Z$ rotations are implemented by the measurement of non-output qubits that have exactly one neighbor. The RL graph with only $YZ$-plane measurements can therefore implement only diagonal unitaries and thus does not allow for universal computation. Compare this with the findings of Ref.~\cite{perdrix_determinism_2017}, where it is proven that MBQC with measurements restricted to the $XZ$-plane (so-called \emph{real MBQC}) that is further performed on a bipartite graph is not universal, because all measurements can be carried out simultaneously. 

\section{Parity Architecture in MBQC}
\label{sec:parity_architecture}

The graph state suitable for uniformly deterministic measurement-based computation, where the outputs coincide with the inputs, is due to the claims proven above necessarily of the register-logic form. This restricted form lends itself to an intuitive interpretation from the computation point of view. Each non-output qubit $u$ in the state is measured in the $YZ$-plane and is attached only to a subset of output qubits $v_1, \ldots, v_n$. The measurement can be decomposed as $M_{YZ}(\theta) = M_Z R_X(-\theta)$, where we use the sign convention of Ref.~\cite{simmons_relating_2021}. The graph-state stabilizer $X_u \otimes Z_{v_1} \otimes \cdots \otimes Z_{v_n}$ of qubit $u$ allows to express the single-qubit rotation $R_X(-\theta)$ as $R_X(-\theta) \coloneqq \exp(-i \theta/2 \, X_u) = \exp(-i \theta/2 \, Z_{v_1} \otimes \cdots \otimes Z_{v_n})$, where the last expression is in general a multiqubit controlled phase gate that acts on the output qubits. To each non-output $u$ one can thus naturally assign a \emph{parity label} \cite{klaver_parity_2026} of the form $v_1 \ldots v_n$ that conveys what outputs are affected by the measurement of $u$. In Fig.~\ref{fig:rl_state}(b), an RL open graph is shown with the parity label of each qubit explicitly depicted. To emphasize the role of each qubit in the state, from now on we refer to an output qubit as a \emph{base qubit} and to a non-output qubit as a \emph{parity qubit}.

Any unitary matrix diagonal in the $Z$ basis can be decomposed into a product of multiqubit controlled phase gates and single-qubit $R_Z$ rotations. Specifically, let $U = \mathrm{diag}(\exp(i \alpha_1), \ldots, \exp(i \alpha_{2^n}))$ be a diagonal matrix for $n$ base qubits. One can understand the $2^n$-tuple $(\alpha_1, \ldots, \alpha_{2^n})$ as a vector in $\mathbb{R}^{2^n}$ and consider the basis of this vector space formed by the rows of the $2^n$-dimensional Hadamard matrix that is constructed using the recursive Sylvester's construction $H_{2^n} \coloneqq H_2 \otimes H_{2^{n-1}}$ \cite{Sylvester1867LXTO}. Each row of $H_{2^n}$ can be seen as the diagonal of a certain tensor product of identity and Pauli $Z$ matrices $Z^{s_1} \otimes \ldots \otimes Z^{s_n}$ for some $s_j \in \{ 0, 1 \}$. The unitary can therefore be decomposed as 
\begin{equation}
    U = \exp \left( i \sum_{s_j} \beta_{s_1, \ldots, s_n} Z^{s_1} \otimes \ldots \otimes Z^{s_n} \right)
\end{equation}
for some real coefficients $\beta_{s_1, \ldots, s_n}$. As follows from above, the unitary can be implemented by a bRL open graph whose form reflects the decomposition just made. The subset of non-zero exponents $s_j$ for a given coefficient $\beta_{s_1, \ldots, s_n}$ identifies the subset of base qubits whose parity the coefficient represents. This in turn corresponds to one parity qubit in a bRL graph. In Fig.~\ref{fig:rl_state}(b), the resulting labeled open graph is shown for a general diagonal unitary acting on four qubits. Note that for the appropriate choice of coefficients, the bRL graph represents the unitary exactly.

\begin{figure*}
    \centering
    \includegraphics[width=.9\linewidth]{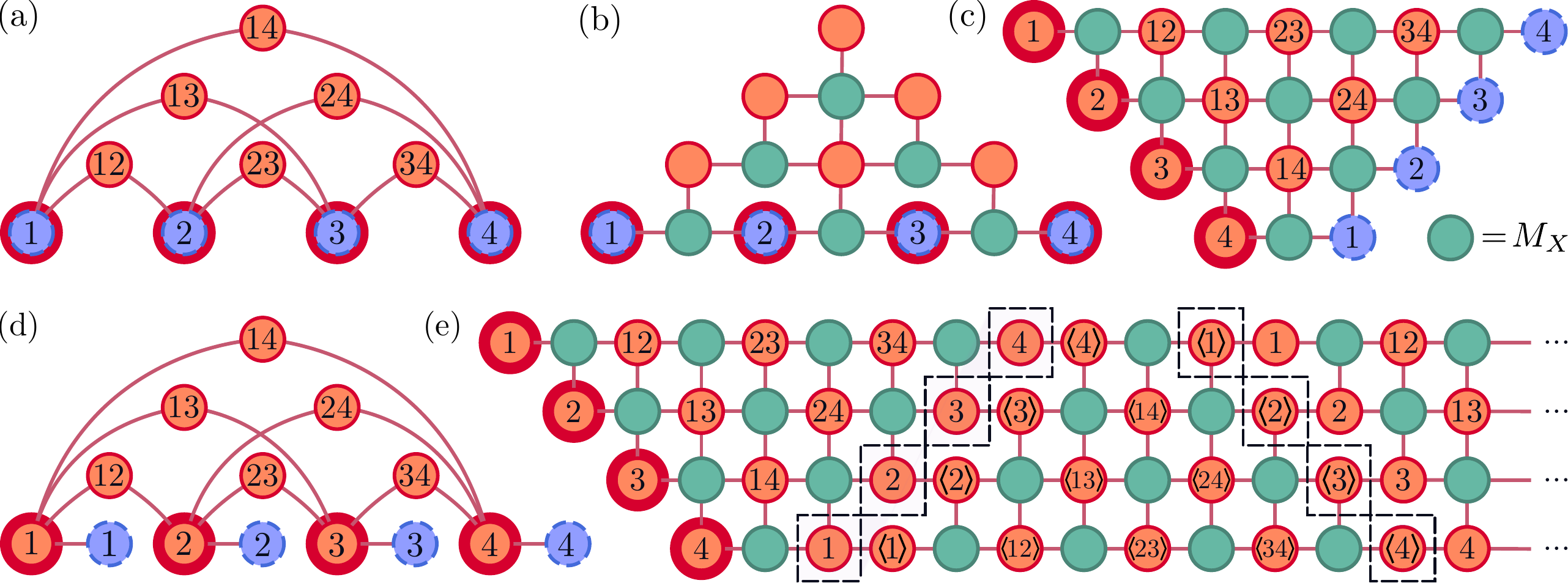}
    \caption{Parity Architecture in MBQC picture. %
    (a) MBQC formulation of the LHZ triangle, where the base qubits are indexed by single-digit numbers. The two-digit numbers indexing each parity qubit represent the two base qubits, whose parity is affected by the local action on that parity qubit. The base qubits are simultaneously both inputs and outputs. Adapted from \cite{smith_parity_2024}. %
    (b) Triangular segment of a cluster state, where every other vertical edge in the last row is intentionally missing. The qubits in the $X$ sublattice, rendered in green (light gray), are measured in Pauli $X$ basis to arrive at the graph in (a). %
    (c) Beveled cluster with four inputs. %
    (d) The post-measurement state for (c) after the Pauli $X$ measurements are made. %
    (e) Series of alternating beveled clusters that supports universal quantum computing. Note the missing vertical edges at the interface of two beveled clusters. Every other beveled cluster implements the $\exp(i \alpha X_i \otimes X_j)$ operations as marked by the parity labels of the form $\avg{ij}$, together with the $R_X$ rotations marked by $\avg{i}$. The graphical convention used throughout is as follows: the input qubits have an extra thick red (dark gray) frame, the outputs are blue with a dashed outline. The qubits that are to be measured in the course of computation are orange (with thin dark gray outline), while those measured in a Pauli basis prior to the computation are green (with thin light gray outline). %
    }
    \label{fig:parity_arch}
\end{figure*}

The identification made above of parity and base qubits fits precisely into the framework of \emph{Parity Architecture}, where parity qubits are manipulated, not necessarily measured, to apply gates to base qubits. This is done by first encoding the state of base qubits into a larger stabilizer quantum state, whose stabilizers are referred to as \emph{plaquettes}. In this encoding, the \emph{Lechner-Hauke-Zoller (LHZ) encoding} \cite{lechner_quantum_2015}, the long-range interactions between the base qubits are mediated by local operations on the parity qubits that are linked by nearest-neighbor interactions only. Originally designed for the annealing platform \cite{lechner_quantum_2015}, the LHZ encoding was later expressed in the language of stabilizers \cite{rocchetto_stabilizers_2016} and also extended to universal digital quantum computing \cite{fellner_universal_2022,fellner_applications_2022}. It was further shown that the use of measurements in this scheme can reduce the encoding and decoding depths \cite{messinger_constant_2023}. This measurement-enhanced modification was then found to be equivalent to a hybrid scheme with an alternating series of gate-based applications of single-qubit operators and measurement-based stages, where each stage is a special case of the scheme in Fig.~\ref{fig:rl_state}(b) \cite{smith_parity_2024}.

The universal hybrid scheme can be fully reformulated in the MBQC picture where multiple stages from Fig.~\ref{fig:rl_state}(b) are ``glued'' onto a series of linear cluster states and where some qubits are measured in the $YZ$-plane and some in the $XY$-plane \cite{smith_minimally_2025}. In general, the resulting graph contains long-range edges and the essential feature of the Parity Architecture is lost. In what follows, a modification of the scheme is presented that supports universal computation and that can be embedded into a graph state with only short-range edges. We focus on the restricted version of the scheme in Fig.~\ref{fig:rl_state}(b), henceforth referred to as the \emph{LHZ triangle}, where only the parity qubits that represent pairwise parities of the base qubits are kept, see Fig.~\ref{fig:parity_arch}(a). A single LHZ triangle can be obtained by starting with a triangular segment of a cluster state, which is bipartite, and measuring one part in the $X$ basis, as shown in Fig.~\ref{fig:parity_arch}(b) \cite{messinger_constant_2023,smith_parity_2024}. The qubits in the other part are measured in the $XY$-plane, but the $X$ measurements incur extra Hadamard gates that effectively turn these into $YZ$-plane measurements. For completeness, in Sec.~\ref{sec:cluster_to_lhz} the explicit series of $X$ basis measurements is demonstrated that transforms the local-interaction cluster into the LHZ triangle together with the measurement plane change. The resulting LHZ triangle is such that its outputs coincide with the inputs, which might not be desirable e.g. when the coherence times of qubits are shorter than the total computational time. A more convenient alternative is the state shown in Fig.~\ref{fig:parity_arch}(c), henceforth referred to as a \emph{beveled cluster}, which features the same pairwise parities as the LHZ triangle. When one sublattice of the beveled cluster is measured in the $X$ basis, the resulting state is the LHZ triangle where each input has an extra neighbor, the corresponding output qubit, cf.~Fig.~\ref{fig:parity_arch}(d). The beveled cluster is the MBQC counterpart of the Parity Twine network \cite{klaver_swap-less_2026,dreier_connectivity-aware_2025}. Note that the triangular cluster in Fig.~\ref{fig:parity_arch}(b) contains $n^2$ physical qubits where $n$ is the number of input (base) qubits, while the beveled cluster for the same number of inputs contains $n^2 + 2 n$ qubits. For the graph state with all pairwise parity qubits and distinct inputs and outputs, this is the minimum number of physical qubits. In order for the addition of a base or parity qubit not to introduce a new degree of freedom, each of them must be accompanied by another qubit as a constraint, and so for $n$ base qubits with $(n^2 - n)/2$ parity qubits and $n$ extra outputs, one needs $2(n + (n^2 - n)/2) + n = n^2 + 2n$ physical qubits \cite{messinger_constant_2023}.

A single beveled cluster can implement any single- and two-qubit unitary diagonal in the $Z$ basis. To reach full (approximate) universality, it suffices to be able to implement any single-qubit $R_X$ rotation. This can be effected by a three-qubit linear cluster attached to each of the outputs of the beveled cluster, where the resulting block of beveled and linear clusters can be tiled one after another to create a sheet of a graph state that supports universal quantum computing. This approach to universality closely follows Ref.~\cite{smith_parity_2024}. Another possibility is to attach multiple beveled clusters directly, where the outputs of the first beveled cluster do \emph{not} coincide with the inputs of the second one. Instead, they are connected by extra edges so that the bulk of the resulting state resembles a standard cluster state with some vertical edges missing. The extra edges effectively turn the pairwise parity labels $ij$ that represent operator $\exp(i \alpha \, Z_i \otimes Z_j)$ into parity labels of the form $\avg{ij}$ that represent $\exp(i \alpha \, X_i \otimes X_j)$ \cite{klaver_parity_2026}. Similarly, the single-qubit parity $\avg{i}$ represents the $R_X$ rotation on qubit $i$. The entire scheme comprises the alternating series of beveled clusters, as shown in Fig.~\ref{fig:parity_arch}(e), where the first beveled cluster implements a unitary $U_Z$ diagonal in the $Z$ basis, the second cluster implements a unitary $U_X$ diagonal in the $X$ basis, the third one implements another $U_Z$ and so on. Obviously, any single-qubit $R_X$ rotation is a special case of $U_X$. 

Each qubit in the series of beveled clusters is measured in the $XY$-plane. There are many that are measured specifically in the $X$ basis and these can be subject to measurement all at once as the preliminary step of the computation. As noted above, the $X$ measurement results in a post-measurement stabilizer state that differs from a graph state by extra Hadamard gates applied to some qubits. These gates can be merged with the $XY$-plane measurements to turn them into $YZ$-plane measurements that act on an actual graph state. In Fig.~\ref{fig:parity_arch}(e), the qubits enclosed in dashed boxes are the only ones that are not affected by this change of the measurement plane. This finishes the discussion on the universal measurement pattern inspired by the Parity Architecture, where after the preliminary step most of the qubits are effectively measured in the $YZ$-plane. In the next section, we demonstrate that $YZ$-plane measurements alone suffice for universality. 

\section{Universal YZ-plane MBQC}
\label{sec:univ_yz}

One of the conclusions drawn in Sec.~\ref{sec:yz_mbqc_with_gflow} is that uniformly deterministic measurement patterns where all non-outputs are measured in the $YZ$-plane do not allow for universal quantum computing. In this section, we depart from the constraints of uniform determinism and study a more general class of deterministic patterns determined by the existence of \emph{Pauli flow} \cite{browne_generalized_2007}. Any measurement in the $YZ$ plane is represented by a measurement basis of the form $M_Z R_X(-\theta)$ for some real $\theta$. This includes not only rotated bases for some non-trivial $\theta$, but also Pauli bases $Y$ and $Z$ that correspond to $\theta = \pi/2$ and $\theta = 0$, respectively. From the point of view of gflow, these two extreme cases play no essential role as the correction strategy has to work for any $\theta$. However, for Pauli flow these cases allow for more general correction strategies. The special standing of the two Pauli bases necessitates the introduction of the following symbolic notation: $\lambda \subseteq YZ$ stands for the case where all non-outputs are measured either in the rotated $YZ$-plane basis, $Y$ basis, or $Z$ basis. It turns out, see Sec.~\ref{sec:pauli_flow_pushed}, that Theorems~\ref{thm:io_g_to_rl_yz} and \ref{thm:mod_isaac_thm_1} are no longer valid when the existence of gflow is replaced by the existence of Pauli flow. Theorem~\ref{thm:rl_yz_impl} does remain valid as any gflow is a special case of Pauli flow, and Lemma~\ref{thm:inputs_gflow} is replaced by:
\begin{lem}[Generalization of Lemma~\ref{thm:inputs_gflow}]
    \label{thm:inputs_pauli_flow}
    Any labeled open graph $(G, I, O, \lambda)$ with Pauli flow and $\lambda \subseteq YZ$ must be such that each input is either also an output or is measured in the $Y$ basis.
\end{lem}
This lemma is proven in Sec.~\ref{sec:pauli_flow_pushed} and represents in the current context the only real restriction on the form of the labeled open graph with Pauli flow. In this section, we demonstrate that one can run universal quantum computation that is driven solely by measurements in the $YZ$ plane of the Bloch sphere. We present two approaches. The first one is the direct demonstration with the explicit construction of Pauli flow for a universal gate set, that can efficiently approximate any unitary, corresponding to \emph{strict universality}. The second approach reduces the problem to the already known result on the \emph{computational universality} of measurements restricted to the $XZ$-plane of the Bloch sphere, where a computationally universal gate set efficiently simulates any unitary from a strictly universal gate set with polylogarithmic overhead in qubits and gate count \cite{aharonov_simple_2003}.

\subsubsection{Strict universality}
\label{sec:univ_yz_pattern}

\begin{figure}
    \centering
    \includegraphics[width=.9\linewidth]{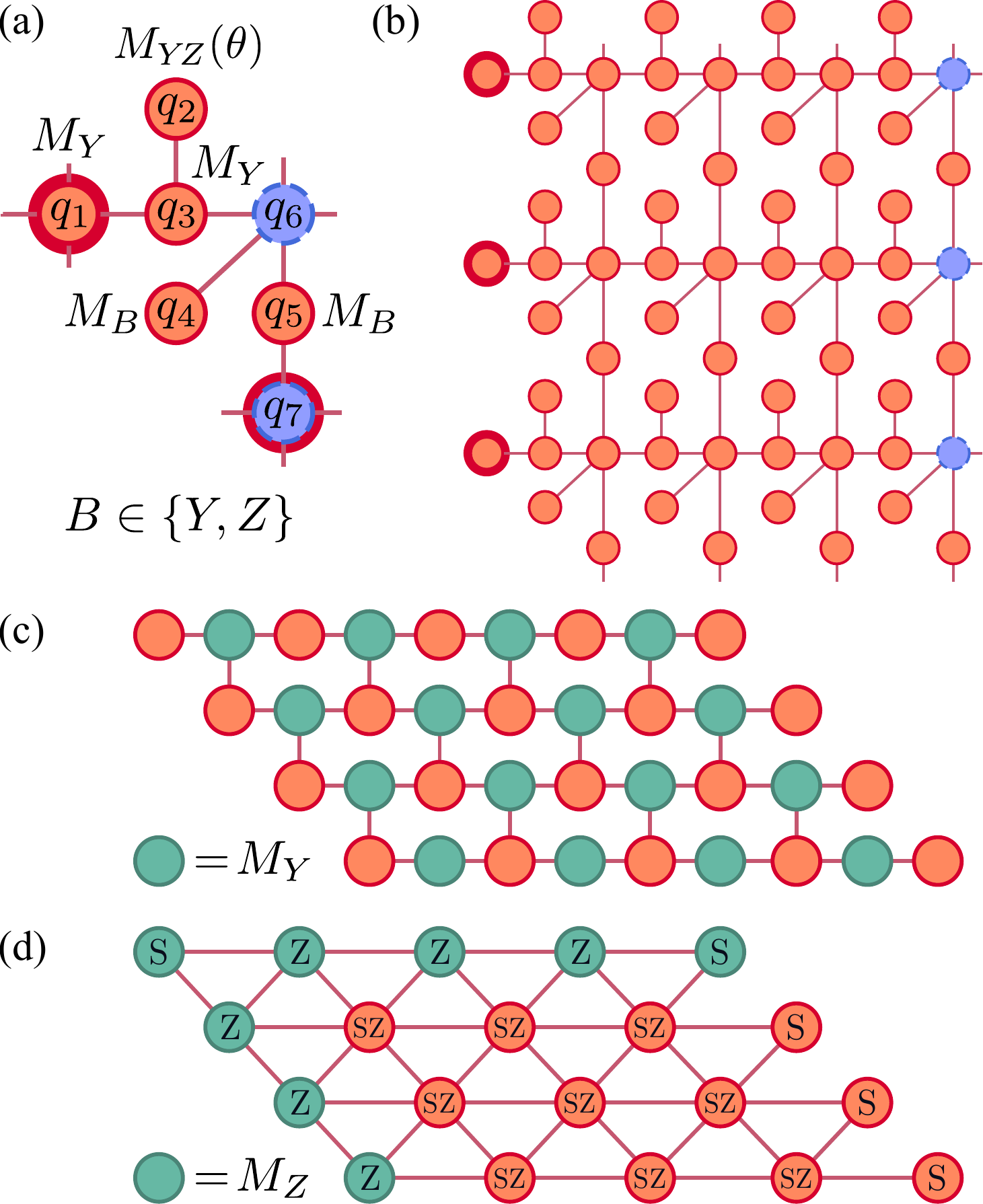}
    \caption{Universal $YZ$-only measurement patterns. %
    (a) The unit cell for complex MBQC with measurement bases and indices depicted. %
    (b) Sheet of unit cells for complex MBQC. %
    (c) Hexagonal grid. One of its two sublattices is subjected to the $Y$ basis measurements. %
    (d) The post-measurement state when all such measurements reduce to $\ket{+i}$ projections. In green (light gray) are marked those qubits that are measured in the $Z$ basis in the second step of the construction of the real-valued universal $YZ$-only pattern. The letters denote the extra gates applied on top of the graph state. %
    }
    \label{fig:univ_yz}
\end{figure}

Here, we present an MBQC scheme that can efficiently approximate any unitary and thus allows for universal quantum computation. Consider a labeled open graph that is constructed by tessellating the unit cell depicted in Fig.~\ref{fig:univ_yz}(a), where the graphical representation follows the notation of Fig.~\ref{fig:parity_arch} except that individual qubits are not marked by their parity label. Instead, they are assigned an index of the form $q_j$ for convenience. Qubit $q_1$ is the first input, qubit $q_6$ is the first output and qubit $q_7$ is both the second input and second output of the unit cell. The center qubit with index $q_3$ is always measured in the $Y$ basis, irrespective of the gate or algorithm implemented and thus can be measured prior to the computation. It is kept there so that each qubit is explicitly measured in the $YZ$-plane. The entire pattern composed of the unit cells is shown in Fig.~\ref{fig:univ_yz}(b).

To demonstrate universality, we show how the unit cell allows for the implementation of the universal set of gates $\{ H, S, \cz, R_Y(\theta) \}$ composed of the generators of the Clifford group together with the class of single-qubit rotations $R_Y$. Note that since the structure of the pattern is regular and fixed, we also have to demonstrate that the unit cell implements the identity. These gates are executed up to Pauli corrections and we discuss each of them in turn in the following. The Pauli corrections are powers of Pauli operators, where the exponents are sums of binary outcomes for the measurement of each qubit. The binary outcome for qubit $q_j$ is denoted by $b_j$ and so $b_j\in\{0,1\}$.

The single-qubit gates require that the bond that connects the two outputs of the unit cell is broken by measuring qubit $q_5$ in the $Z$ basis. Qubits $q_6$ and $q_7$ acquire an extra $Z$ correction $Z^{b_5}$. For the case of a single-qubit gate $G$ acting on the input we thus focus only on the upper part of the cell composed of qubits $q_1$, $q_2$, $q_3$, $q_4$, and $q_6$. 
When the measurement basis of qubit $q_4$ is set to $Y$, the pattern implements the gate $G_Y(\theta) = H R_Y((-1)^{1 + b_1} \theta)$ with Pauli correction on qubit $q_6$ of the form $X^{b_1+b_2+b_3} Z^{1+b_2+b_3+b_4+b_5}$. When instead the measurement basis for qubit $q_4$ is set to $Z$, one gets $G_Z(\theta) = S G_Y(\theta)$ with Pauli correction on qubit $q_6$ of the form $X^{b_1+b_2+b_3} Z^{1 + b_1+b_4+b_5}$. It is easy to verify that $R_Y((-1)^c \pi/2) = Z^{1+c} X^c H$ and so we can compile the following table, where for a given gate $G$ the measurement basis $M^{(4)}$ of qubit $q_4$, the measurement angle $\theta^{(2)}$ of qubit $q_2$, as well as the corresponding Pauli correction is listed: 
\begin{center}
    \begin{tabular}{c|c|c|l}
    $G$ & $M^{(4)}$ & $\theta^{(2)}$ & Correction \\ \hline
    $\ident$ & $M_Y$ & $\pi/2$ & $X^{b_2+b_3} Z^{b_1+b_2+b_3+b_4+b_5}$ \\
    $S$ & $M_Z$ & $\pi/2$ & $X^{b_2+b_3} Z^{b_1+b_4+b_5}$ \\
    $H$ & $M_Y$ & $0$ & $X^{b_1+b_2+b_3} Z^{1+b_2+b_3+b_4+b_5}$ \\
    $H R_Y(\theta)$ & $M_Y$ & $(-1)^{1 + b_1} \theta$ & $X^{b_1+b_2+b_3} Z^{1+b_2+b_3+b_4+b_5}$
\end{tabular}
\end{center}
The Pauli corrections entering the gates are not considered above but it is easy to adjust the formulas accordingly.

For the two-qubit gate, we can choose the pattern for the identity on the first logical qubit discussed above and measure qubit $q_5$ in the $Y$ basis. It turns out that such a pattern leads to a controlled phase gate $\cz_{12} (S_1 \otimes S_2)$ with Pauli correction $X_1^{b_2 + b_3} Z_1^{1 + b_1 + b_4 + b_5} \otimes Z_2^{1 + b_2 + b_3 + b_5}$. This finishes the proof that the measurement pattern that is formed by the tessellation of the unit cell in Fig.~\ref{fig:univ_yz}(a) supports universal quantum computation, where all qubits are measured in the $YZ$-plane only. As follows from the results of Ref.~\cite{smith_parity_2024} the pattern just presented does not have any gflow. However, it has a Pauli flow. To see that, it is enough to consider only a single column of the pattern, formed by stitching multiple cells vertically on top of each other, where the cell at the bottom lacks qubits $q_5$ and $q_7$, cf.~Fig.~\ref{fig:univ_yz}(b). This is allowed as one can think of the entire computation as being performed in separate steps (corresponding to separate time steps in the circuit picture), where each step is implemented by a single column of the pattern. Even more so, each unit in each column can be performed in parallel and so it suffices to consider the Pauli flow of a single cell. One can check that the following choice of correction function $p$ and strict partial order $\prec$ satisfies the definition of Pauli flow in Def.~\ref{def:pauli_flow}: $p(q_1) = \{ q_3 \}$, $p(q_2) = \{ q_2 \}$, $p(q_3) = \{ q_6 \}$, $p(q_4) = \{ q_4 \}$, and $p(q_5) = \{ q_5 \}$ with the strict partial order given by $q_1 \prec q_2 \prec q_3 \prec q_4, q_5, q'_5 \prec q_6, q_7$. 
Qubits $q_6$ and $q_7$ are unmeasured outputs and $q'_5$ denotes the qubit $q_5$ in the unit cell lying above the one considered. It is not hard to see that including this qubit from another cell does not affect the Pauli flow for the unit cell in question.

The standard literature concludes that all Pauli measurements can be performed all at once as the very first step of the computation \cite{raussendorf_computational_2002}. The corresponding correction strategy is presented in Sec.~\ref{sec:pauli_flow_pushed}.

\subsubsection{Computational universality}
\label{sec:comp_univ}

Measurements restricted to the $YZ$-plane are universal not only in the strict sense discussed in the previous section, but also allow for a weaker notion of universality, computational universality, where real-valued gates implemented by the measurements are in turn used to simulate a universal set of complex-valued gates. Even though the measurements in the $YZ$ plane do include the imaginary unit and the post-measurement state can be complex, there exists a class of patterns with $YZ$-plane measurements that can be mapped to a class of patterns that are known to be universal for real-valued computation and that employ only $XZ$-plane measurements.

Specifically, consider a hexagonal grid depicted in Fig.~\ref{fig:univ_yz}(c) that is bipartite and measure all qubits in one of the sublattices in the $Y$ basis. Following the rules for the $Y$ basis measurement \cite{hein_entanglement_2006,hein_multi-party_2004}, the post-measurement state can be represented by a graph state with the graph being a triangular grid, while individual surviving qubits get an additional local gate $Z$, $S$, or $SZ$ that acts on them. In the second step, shown in Fig.~\ref{fig:univ_yz}(d), one measures all qubits in the first column and first row of the state in the $Z$ basis, which merely removes them without affecting the graph topology. The remaining qubits are acted on either by $S Z$ or by $S$, depending on the measurement outcome and whether the qubit is an output or not. These extra gates can be merged with the subsequent measurement as $M_{YZ}(\theta) S = M_{Z} S S^\dagger R_X(-\theta) S = M_{Z} S R_Y(\theta) = M_{Z} R_Y(\theta) = M_{XZ}(\theta)$. Similarly, $M_{YZ}(\theta) S Z = M_{XZ}(-\theta)$ and so all $YZ$-plane measurements turn into $XZ$-plane measurements. For the latter it was shown that the triangular grid allows for the universal computation, where real unitaries thus implemented simulate a universal set of gates \cite{mhalla_graph_2012}. This final step is for completeness demonstrated explicitly in Sec.~\ref{sec:explicit_xz_univ}. Note that we keep all the physical qubits initially stabilized in $\ket{+}$. The input state is produced in the preparation stage of the computation that starts from these stabilized states and then applies real-valued gates.

As a final remark, there are only three regular tilings of the two-dimensional plane \cite{van_den_nest_universal_2006} and our results show that to each of them there is naturally a principal plane of the Bloch sphere assigned for which the corresponding MBQC is universal: the rectangular tiling when understood as the underlying graph of the resource graph state is universal for $XY$-plane-only MBQC \cite{mantri_universality_2017}. The triangular tiling is universal for $XZ$-plane-only (real) MBQC \cite{mhalla_graph_2012}, and finally the hexagonal tiling is universal for $YZ$-plane-only MBQC. Note that when no restrictions on the measurement plane are made the three tilings were shown to be universal in Ref.~\cite{van_den_nest_universal_2006}.

\section{Conclusion}
\label{sec:conclusion}

We extend the discussion on the measurement patterns with measurements restricted to the $YZ$-plane of the Bloch sphere and study its relations to the Parity Architecture. For the former, the exhaustive discussion is performed for the patterns with either gflow or Pauli flow in terms of four relevant features, such as the form of the underlying graph or whether the inputs coincide with the outputs. We improve some results of Ref.~\cite{smith_parity_2024} and demonstrate among others that any uniformly deterministic measurement pattern whose inputs coincide with the outputs must necessarily have all the non-output qubits measured in the $YZ$-plane and its underlying graph must be of the register-logic form, defined in the text. The measurement plane for such patterns is therefore dictated by the mere requirement that the computation is uniformly deterministic. Departing from the uniform determinism, we show in two different ways that universal quantum computation is possible when only $YZ$-plane measurements are utilized. This work can therefore be seen as closing the line of research initiated by Ref.~\cite{mantri_universality_2017} for universal $XY$-plane-only MBQC and continued by Ref.~\cite{mhalla_graph_2012} for universal $XZ$-plane-only MBQC. A possible avenue for further research is to consider even more general notions of determinism \cite{mhalla_shadow_2025}.

The register-logic graphs with $YZ$-plane measurements fit naturally into the framework of the Parity Architecture. The non-outputs in the graph represent parities of the outputs and the LHZ encoding can be used to embed the register-logic graph with pairwise parities into a larger graph that features only short-range interactions. The resulting layout fits into planar chip designs or, as is the case for photonic platforms, requires single-photon sources with only nearest-neighbor connectivity, compare the layouts in Refs.~\cite{bourassaBlueprintScalablePhotonic2021,Bartolucci2023,Gliniasty2024}. Translating some of the ideas of the Parity Twine framework \cite{klaver_parity_2026,dreier_connectivity-aware_2025} into the MBQC domain leads to the introduction of the beveled cluster. While the beveled cluster still possesses only short-range edges, its inputs are distinct from the outputs and so the inputs can be measured in the initial stages of the computation. This might be desirable for hardware platforms with short coherence times where qubits cannot be kept in memory throughout the entire run of the quantum algorithm. For photonic platforms, for example, reducing the time from the source to the measurement increases the robustness of the platform to photon loss \cite{Bartolucci2023}. Multiple beveled clusters can be concatenated to create a measurement pattern that keeps the aforementioned desirable features and allows for universal computation. Though lying beyond the scope of the present work, the Parity Twine framework can be translated in its entirety into the MBQC picture. This opens the door to studying the resource-optimization properties of the framework within the measurement-based domain. 

\emph{Acknowledgements:} The authors thank Mehdi Mhalla and Isaac Smith for valuable discussions, and Berend Klaver and Christophe Goeller for insightful comments on the manuscript. This study was supported by the Austrian Research Promotion Agency (FFG Project No. FO999937388, FFG Basisprogramm).

\bibliography{refs}

\appendix

\section{MBQC overview}
\label{sec:prelims}

The measurement-based quantum computation (MBQC) in the formulation of the one-way computer \cite{raussendorf_one-way_2001} consists of single-qubit measurements done on a highly-entangled quantum state of many qubits. This so-called \emph{graph state} is characterized by the underlying graph $G = (V, E)$, where $V$ is the set of qubits playing the role of nodes, and $E$ is the set of edges that determine the pairs of qubits that are subjected to entangling $\cz$ gates. The graph state is then a stabilizer state defined as the unique eigenstate $\ket{\psi}$ of the equations $K_u \ket{\psi} = \ket{\psi}$ for each $u \in V$, where $K_u = X_u \prod_{w \in \neigh{u}} Z_w$ is the \emph{elementary stabilizer} of qubit $u$ and where by $\neigh{u}$ one denotes the \emph{neighborhood} of qubit $u$, which does not contain $u$.

To account for the fact that the input qubits of the computation might be in an arbitrary state, one introduces the \emph{open graph}, which is a triple $(G, I, O)$, where $I \subseteq V$ and $O \subseteq V$ are the sets of input and output qubits, respectively. Note that the two sets $I$ and $O$ need not be disjoint and that we always assume that each connected component of $G$ contains at least one input or output. The qubits in $I$ are not associated with an elementary stabilizer and the quantum state given by the open graph, which we still refer to as a graph state, is defined by stabilizers $\{ K_u \}_{u \in \Bar{I}}$ of only the non-input qubits $\Bar{I} \coloneqq V \setminus I$. Each non-output qubit $u \in \Bar{O}\coloneqq V \setminus O$ is in the course of the computation measured in some basis that is either a Pauli basis, or a rotated basis that lies in one of the principal planes of the Bloch sphere. For a more convenient description, the notion of \emph{labeled open graph} is defined as the quadruple $(G, I, O, \lambda)$, where $\lambda: \Bar{O} \to \{ X, Y, Z, XY, XZ, YZ \}$ is a mapping that to each measured qubit assigns the measurement plane or Pauli basis. The measurement-based execution of a particular algorithm is characterized by a \emph{measurement pattern}, which is a sequence of several types of commands, essentially describing the creation of the initial graph state, subsequent adaptive single-qubit measurements as well as the correction Pauli operators. For a more rigorous definition see e.g.~Ref.~\cite{backens_there_2021}.

Since Pauli gates are involutive, i.e., an even number of, say, $Z$ gates meeting at the same qubit cancel each other out, one cares only about an odd number of applications of a Pauli operator. This leads to the notion of the \emph{odd neighborhood}, which is for a given set of nodes $S \subseteq V$ defined as
\begin{equation}
    \nodd{S} = \{ v \in V: |\neigh{v} \cap S| = 1 \pmod{2} \},
\end{equation}
where $|\neigh{v} \cap S|$ stands for the number of elements of set $\neigh{v} \cap S$. Note that the odd neighborhood $\nodd{S}$ does \emph{not} have to be disjoint with $S$. Further note that for any two subsets $A, B \subseteq V$, the symmetric difference $\Delta$ of their odd neighborhoods satisfies $\nodd{A \Delta B} = \nodd{A} \Delta \nodd{B}$.

Each single-qubit measurement leads to two possible outcomes, one of which is considered desirable since no subsequent correction is necessary. The other, undesirable, outcome entails a correction Pauli operator that has to be applied to some not-yet-measured qubits. The exact strategy how to apply the correction operators is embodied by various notions of \emph{flow}, some of which are defined below. In the following definitions, $\mathcal{P}(\Bar{I})$ stands for the power set of the set of non-input qubits $\Bar{I}$.

\begin{defi}[\textbf{Gflow}; adapted from \cite{backens_there_2021}\label{def:gflow}]
A labeled open graph $(G, I, O, \lambda)$ has \emph{generalised flow} (or \emph{gflow}) if there exists a map $g: \Bar{O} \to \mathcal{P}(\Bar{I})$ and a strict partial order $\prec$ over $V$ such that for all $u \in \Bar{O}$ the following properties hold:
\begin{enumerate}[a)]
    \item $\forall v \in g(u): u \neq v \implies u \prec v$ \label{itm:g1}
    \item $\forall v \in \nodd{g(u)}: u \neq v \implies u \prec v$ \label{itm:g2}
    \item $\lambda(u) = XY \implies u \not\in g(u) \wedge u \in \nodd{g(u)}$ \label{itm:lam_xy}
    \item $\lambda(u) = XZ \implies u \in g(u) \wedge u \in \nodd{g(u)}$
    \item $\lambda(u) = YZ \implies u \in g(u) \wedge u \not\in \nodd{g(u)}$ \label{itm:lam_yz}
\end{enumerate}
\end{defi}

Note that the existence of gflow for a given labeled open graph is equivalent to the ability to run a measurement pattern on that graph that is strongly, uniformly, and step-wise deterministic \cite{browne_generalized_2007}. In the main text, we focus on the uniform part of determinism but the other two aspects, strong and step-wise determinism, are still assumed. When the requirements on the type of determinism are relaxed, more general correction strategies are allowed that are captured by the existence of Pauli flow \cite{browne_generalized_2007}, defined next.

\begin{defi}[\textbf{Pauli flow}; adapted from \cite{simmons_relating_2021,mcelvanney_complete_2023}\label{def:pauli_flow}]
    Given a labeled open graph $\Gamma = (G, I, O, \lambda)$, a \emph{Pauli flow} for $\Gamma$ is a tuple $(p, \prec)$ of a map $p: \Bar{O} \to \mathcal{P}(\Bar{I})$ and a strict partial order $\prec$ over $V$ such that for all $u \in \Bar{O}$ the following properties hold:
    \begin{enumerate}[a)]
        \item $\forall v \in p(u)\setminus\{u\}: \lambda(v) \not\in \{ X, Y \} \implies u \prec v$ \label{itm:pauli_flow_corr_order}
        \item $\forall v \in \nodd{p(u)}\setminus\{u\}: \lambda(v) \not\in \{ Y, Z \} \implies u \prec v$ \label{itm:pauli_flow_nodd_order}
        \item $\forall v \in V\setminus\{u\}: \lambda(v) = Y \wedge \neg (u \prec v) \implies (v \in p(u) \Leftrightarrow v \in \nodd{p(u)})$ \label{itm:pauli_flow_p_iff_oddp}
        \item $\lambda(u) = XY \implies u \not\in p(u) \wedge u \in \nodd{p(u)}$ \label{itm:pauli_flow_xy}
        \item $\lambda(u) = XZ \implies u \in p(u) \wedge u \in \nodd{p(u)}$ \label{itm:pauli_flow_xz}
        \item $\lambda(u) = YZ \implies u \in p(u) \wedge u \not\in \nodd{p(u)}$ \label{itm:pauli_flow_yz}
        \item $\lambda(u) = X \implies u \in \nodd{p(u)}$ \label{itm:pauli_flow_x}
        \item $\lambda(u) = Y \implies (u \not\in p(u) \wedge u \in \nodd{p(u)}) \lor (u \in p(u) \wedge u \not\in \nodd{p(u)})$ \label{itm:pauli_flow_y}
        \item $\lambda(u) = Z \implies u \in p(u)$ \label{itm:pauli_flow_z}
    \end{enumerate}
\end{defi}

As discussed in Ref.~\cite{mhalla_finding_2008}, for any given open graph and its gflow, one can partition the qubits in $V$ into pairwise disjoint subsets $V_k^\prec$ that cover $V$, such that
\begin{align}
    V_0^\prec & \coloneqq \mathrm{max}_\prec (V), \\
    V_{k+1}^\prec & \coloneqq \mathrm{max}_\prec (V \setminus \Tilde{V}^\prec_{k}), \quad k \geq 0, \label{eq:v_layers}
\end{align}
where $\mathrm{max}_\prec (X) = \{ u \in X | \forall v \in X: \neg (u \prec v) \}$ for set $X$, and $\Tilde{V}^\prec_k \coloneqq \bigcup_{i=0,\ldots,k} V^\prec_i$. Then for a given open graph and its two given gflows $(g, \prec)$ and $(g', \prec')$, one says that $(g, \prec)$ is \emph{more delayed} than $(g', \prec')$ if $|\Tilde{V}^\prec_k| \geq |\Tilde{V}^{\prec'}_k|$ for all $k$ and there exists $k_0$ for which $|\Tilde{V}^\prec_{k_0}| > |\Tilde{V}^{\prec'}_{k_0}|$ \cite{mhalla_finding_2008}. With that we can finally state the following definition.

\begin{defi}[\textbf{Maximally delayed gflow}, \cite{mhalla_finding_2008}]
    A gflow $(g, \prec)$ is \emph{maximally delayed} if there exists no gflow of the same open graph that is more delayed.
\end{defi}

Note that if a labeled open graph has a gflow, it also has a maximally delayed gflow \cite{mhalla_finding_2008,backens_there_2021}. Since we assume that every connected component of the open graph contains at least one input or output qubit, by the first part of the proof of Lemma~C.2 in Ref.~\cite{backens_there_2021} any gflow satisfies $V_0^\prec \subseteq O$ with equality $V_0^\prec = O$ holding for a maximally delayed gflow. For $V_1^\prec$, the following lemma applies:

\begin{lem}[Lemma~C.4 of Ref.~\cite{backens_there_2021}]
    \label{thm:lemma_c_4}
    If $(g, \prec)$ is a maximally delayed gflow of a labeled open graph $(G, I, O, \lambda)$, then $V^\prec_1 = V^{\prec, XY}_1 \cup V^{\prec, XZ}_1 \cup V^{\prec, YZ}_1$, where
    \begin{align}
        V^{\prec, XY}_1 & = \{ u \in \Bar{O}| \lambda(u) = XY \wedge \exists K \subseteq O \setminus I: \nonumber \\
        & \qquad \qquad \qquad \nodd{K} \setminus O = \{ u \} \}, \\
        V^{\prec, XZ}_1 & = \{ u \in \Bar{O}| \lambda(u) = XZ \wedge \exists K \subseteq O \setminus I: \nonumber \\
        & \qquad \qquad \qquad \nodd{K \cup \{ u \}} \setminus O = \{ u \} \}, \\
        V^{\prec, YZ}_1 & = \{ u \in \Bar{O}| \lambda(u) = YZ \wedge \exists K \subseteq O \setminus I: \nonumber \\
        & \qquad \qquad \qquad \nodd{K \cup \{ u \}} \setminus O = \emptyset \}.
    \end{align}
\end{lem}

\section{$YZ$-plane patterns with gflow}
\label{sec:cond_discussion}

In this section we prove the claims of Sec.~\ref{sec:yz_mbqc_with_gflow} and study the mutual relation of the four features introduced there. The full discussion is made of these features with explicit counterexamples presented where applicable. 

\begin{thm}[Restatement of Thm.~\ref{thm:io_g_to_rl_yz}]
    Let $(G, I, O, \lambda)$ be a labeled open graph with $I = O$ and with gflow. Then $(G, I, O)$ is of the register-logic form and all non-outputs are measured in the $YZ$-plane.
\end{thm}
\begin{proof}
    If the pattern has gflow, then it also has a maximally delayed gflow $(g, \prec)$. From Lemma~\ref{thm:lemma_c_4} the condition $I = O$ implies that the only subset $K \subseteq O \setminus I$ is the empty set, for which $\nodd{K} \setminus O  = \emptyset$ and thence $V^{\prec, XY}_1 = \emptyset$. Similarly, $V^{\prec, XZ}_1 = \emptyset$ since $\nodd{\emptyset \cup \{ u \}} = \neigh{u}$ and the neighborhood of $u$ does not contain $u$ by definition. Hence, $V^{\prec}_1 = V^{\prec, YZ}_1 = \{ u \in \Bar{O} | \lambda(u) = YZ \wedge \neigh{u} \setminus O = \emptyset \}$. All qubits in $V^{\prec}_1$ are therefore such that their neighborhood fully lies in $O$. One can now use Lemma~C.3 in Ref.~\cite{backens_there_2021} and consider the restriction of $g$ to $V \setminus \tilde{O}$, where $\tilde{O} \coloneqq O \cup V_1^\prec$ is the new set of outputs, which yields again a maximally delayed gflow $(\tilde{g},\tilde{\prec})$. Applying Lemma~\ref{thm:lemma_c_4} one more time lets us compute $V^\prec_2$ by considering the three conditions as follows: 
    \begin{align}
        \nodd{K} \setminus \tilde{O} = \{ u \}, \\
        \nodd{K \cup \{ u \}} \setminus \tilde{O} = \{ u \}, \\
        \nodd{K \cup \{ u \}} \setminus \tilde{O} = \emptyset,
    \end{align}
    for $K \subseteq V^\prec_1$ and $u \in V \setminus \tilde{O}$. To see that the first condition never takes place, recall that $\nodd{K} \subseteq K \cup \bigcup_{v \in K} \neigh{v} \subseteq \tilde{O}$ and so $\nodd{K} \setminus \tilde{O} = \emptyset$. Furthermore, $u \not\in \tilde{O}$ and so also $u \not\in K$. One can therefore write $\nodd{K \cup \{ u \}} = \nodd{K \Delta \{ u \}} = \nodd{K} \Delta \nodd{ \{ u \} } = \nodd{K} \Delta \neigh{u}$, where again $\nodd{K} \subseteq \tilde{O}$ and so $\nodd{K \cup \{ u \}} \setminus \tilde{O} = \neigh{u} \setminus \tilde{O}$. The right-hand side of the last equation does not contain $u$ and so the second condition above cannot be satisfied. Finally, the third condition can be satisfied only when $\neigh{u} \subseteq \tilde{O}$ and so $u$ can have neighbors among the outputs as well as qubits in $V^\prec_1$. But any neighbors of qubits in $V^\prec_1$ lie in $O$ and so must $u$, which is a contradiction. The only possibility left is thus that $\neigh{u} \subseteq O$, where $\lambda(u) = YZ$. But all such qubits $u$ are already included in $V^\prec_1$, which is again a contradiction. To conclude, $V^\prec_2 = \emptyset$. Note that if $V^{\prec}_k = \emptyset$ for some $k$, then $V^{\prec}_l = \emptyset$ for all $l \geq k$. This implies that all other sets $V^{\prec}_k = \emptyset$ for $k \geq 2$ and the entire graph consists only of the outputs $O = V^\prec_0$ and the nodes that form $V^\prec_1$ and that are all measured in $YZ$, while their neighborhoods are fully contained in $O$. The open graph is therefore of the register-logic form. This conclusion is obviously valid irrespective of the fact that the maximally delayed gflow has been used instead of the original gflow. Note that this proof, specifically the use of Lemma~\ref{thm:lemma_c_4}, is the only place in this work where we use the assumption that any connected component of $G$ has at least one input or output.
\end{proof}

\begin{thm}[Restatement of Thm.~\ref{thm:rl_yz_impl}]
    Let $\Gamma = (G, I, O, \lambda)$ be a labeled open graph with $I \subseteq O$, where $(G, I, O)$ is of the register-logic form and all non-outputs are measured in the $YZ$-plane. Then $\Gamma$ has gflow.
\end{thm}
\begin{proof}
    The proof follows essentially the first part of the proof of Theorem~1 in the supplementary of Ref.~\cite{smith_parity_2024}, where the condition of the graph being bipartite with respect to $I$ is replaced by the RL form and the condition $I \subseteq O$ is assumed on top. Namely, let us define $g(v) \coloneqq \{ v \}$ for each $v \in \Bar{O}$ and $v \prec u$ iff $u \in O$. Then $g: \Bar{O} \to \mathcal{P}(\Bar{I})$ as required and $\prec$ is obviously a strict partial order. Furthermore, $(g, \prec)$ satisfies condition \ref{itm:g1}) of the gflow definition vacuously. Since all neighbors of $v$ lie in $O$, it follows that $\nodd{g(v)} = \neigh{v} \subseteq O$ and so $v \not\in \nodd{g(v)}$. Condition \ref{itm:lam_yz}) is thus satisfied and so is condition \ref{itm:g2}) as for each $u \in \nodd{g(v)}$ it holds that $v \prec u$. The pair $(g, \prec)$ is therefore a valid gflow.
\end{proof}

\begin{figure}
    \centering
    \includegraphics[width=0.95\linewidth]{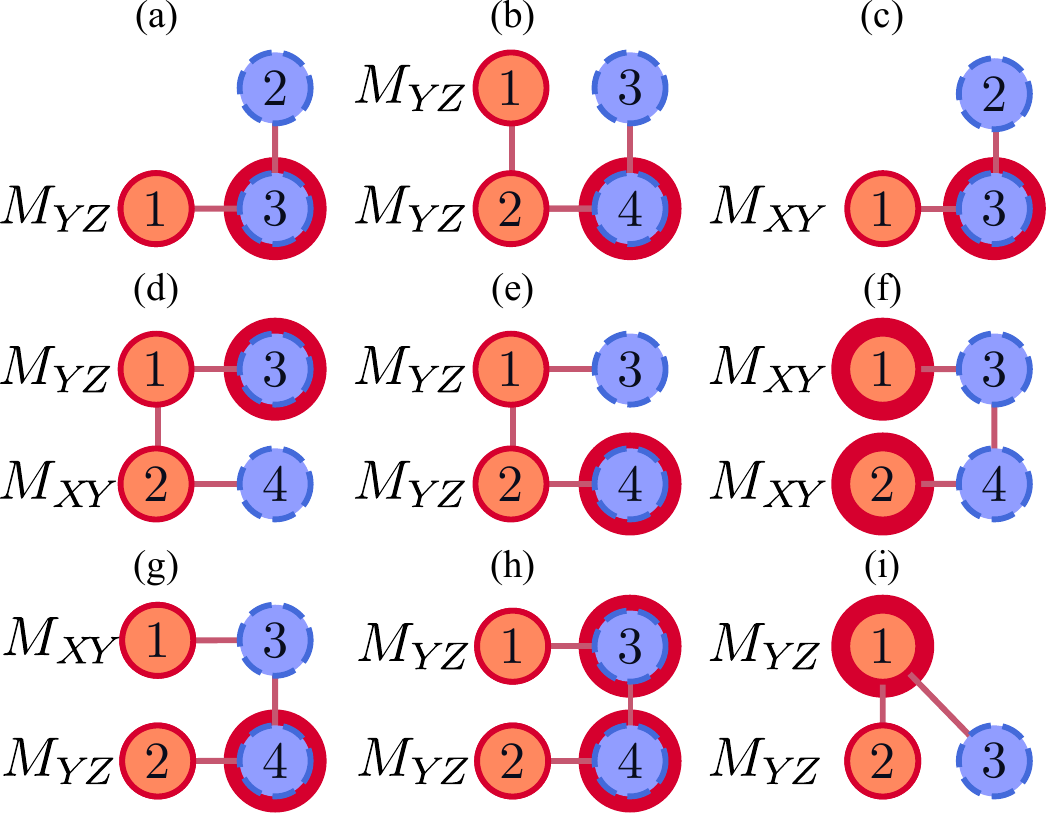}
    \caption{%
    Examples and counterexamples of labeled open graphs. The graphics follows the convention of Fig.~\ref{fig:parity_arch} except that the number in each qubit is merely an index, not a parity label. %
    (a) Graph with $I \subsetneq O$ and gflow given by $g(1) = \{ 1 \}$ and $1 \prec 3$. %
    (b) Graph with $I \subsetneq O$ that is not of the RL form and has no gflow. One can also remove qubit 3 and consider a similar graph with $I = O$, which also has no gflow. %
    (c) Graph of the RL form with $I \subseteq O$ and no gflow. A valid gflow would have to satisfy $3 \in g(1)$ which is not allowed as qubit 3 is an input. The same conclusion holds when considering the same graph but with qubit 2 removed and so $I = O$. %
    (d) Graph with $I \subsetneq O$ and gflow given by $g(2) = \{4\}$, $g(1) = \{1\}$ with $1 \prec 2 \prec 3,4$. %
    (e) Graph with $I \subsetneq O$ and gflow that is not of the RL form. The gflow reads $g(1) = \{ 1 \}$, $g(2) = \{ 2,3 \}$ with $1 \prec 2 \prec 3, 4$. %
    (f) Graph of the RL form with gflow, where $I \cap O = \emptyset$. The gflow reads $g(1) = \{ 3 \}$, $g(2) = \{ 4 \}$ with $1,2 \prec 3,4$. %
    (g) Pattern with $I \subsetneq O$ and gflow, where some maximal elements are not measured in the $YZ$-plane. The gflow reads $g(1) = \{3\}$, $g(2) = \{2\}$ with $1,2 \prec 3,4$. %
    (h) Graph with $I = O$ and gflow that is not bipartite with $I$ forming one part. The gflow is given by $g(1) = \{1\}$, $g(2) = \{2\}$ with $1, 2 \prec 3, 4$. %
    (i) Bipartite graph with $|I| = |O|$ and $I$ forming one part, yet with no gflow as the inputs that are not also outputs must be measured in the $XY$-plane. %
    }
    \label{fig:counter_gflow}
\end{figure}

Note that in Theorem~\ref{thm:rl_yz_impl}, in general the inputs form only a subset of the outputs. In Fig.~\ref{fig:counter_gflow}(a) there is an example of a labeled open graph that complies with the theorem and that has $I \subsetneq O$. The same example illustrates that Lemma~\ref{thm:inputs_gflow} really implies $I \subseteq O$, not necessarily the equality of the two sets.

As mentioned in the main text, one can identify four features of labeled open graphs $(G, I, O, \lambda)$ that are relevant when discussing $YZ$-only patterns. Namely: $I \subseteq O$; $\lambda \equiv YZ$; $(G, I, O)$ is of the RL form; and $(G, I, O, \lambda)$ has gflow. It is easy to convince oneself that any of the four conditions alone does not imply any of the others. Let us discuss in the following if any pair of these conditions imply the other two, and if any triple of them implies the remaining one. For the pairs we obtain these conclusions:
\begin{itemize}
    \item Assuming $I \subseteq O \wedge \lambda \equiv YZ$: Neither of the two remaining conditions is implied; for two counterexamples with $I = O$ and $I \subsetneq O$ refer to Fig.~\ref{fig:counter_gflow}(b). In both cases, for a valid gflow it must be that qubit 2 satisfies $2 \in g(2)$ due to condition~\ref{itm:lam_yz}). But then necessarily $1 \in \nodd{g(2)}$ irrespective of the exact form of $g(2)$ and so $2 \prec 1$ due to condition~\ref{itm:g2}). At the same time, for qubit 1 it must analogously hold that $1 \in g(1)$ and since the inputs cannot lie in correction sets then also $4 \not\in g(1)$. This implies that $2 \in \nodd{g(1)}$ and condition~\ref{itm:g2}) leads to $1 \prec 2$, which is a contradiction and so no gflow exists. %
    \item Assuming $I \subseteq O \wedge RL$: Neither $\lambda \equiv YZ$ nor the existence of gflow is implied; for two counterexamples with $I = O$ and $I \subsetneq O$ see Fig.~\ref{fig:counter_gflow}(c).
    \item Assuming $I \subseteq O \wedge \text{gflow}$: For $I = O$ Theorem~\ref{thm:io_g_to_rl_yz} implies that both conditions $\lambda \equiv YZ$ and $RL$ are satisfied. For $I \subsetneq O$ this is no longer true as witnessed by the counterexample in Fig.~\ref{fig:counter_gflow}(d).
    \item Assuming $\lambda \equiv YZ \wedge RL$: Neither $I \subseteq O$ nor the existence of gflow is implied. For a counterexample consider a labeled open graph of the RL form where the non-outputs are measured in $YZ$, while the inputs form a subset of $\Bar{O}$. This pattern obviously does not satisfy $I \subseteq O$ and since the inputs for a graph with gflow must be measured in $XY$-plane, see Lemma~\ref{thm:inputs_gflow}, the graph just considered has no gflow.
    \item Assuming $\lambda \equiv YZ \wedge \text{gflow}$: Lemma~\ref{thm:inputs_gflow} implies condition $I \subseteq O$. However, condition $RL$ does not have to hold as demonstrated in the counterexample in Fig.~\ref{fig:counter_gflow}(e).
    \item Assuming $RL \wedge \text{gflow}$: Neither $I \subseteq O$ nor $\lambda \equiv YZ$ is implied; for a counterexample see Fig.~\ref{fig:counter_gflow}(f).
\end{itemize}
The pair-wise conditions are in most cases too weak to imply the others. For the triples the situation somewhat differs:
\begin{itemize}
    \item Assuming $\lambda \equiv YZ \wedge RL \wedge \text{gflow}$: Condition $I \subseteq O$ is implied by Lemma~\ref{thm:inputs_gflow}.
    \item Assuming $I \subseteq O \wedge \lambda \equiv YZ \wedge RL$: The existence of gflow is implied by Theorem~\ref{thm:rl_yz_impl}.
    \item Assuming $I \subseteq O \wedge \lambda \equiv YZ \wedge \text{gflow}$: For $I = O$ this triple of assumptions implies condition $RL$ due to Theorem~\ref{thm:io_g_to_rl_yz}. However, for $I \subsetneq O$ this implication does not have to hold, see Fig.~\ref{fig:counter_gflow}(e).
    \item Assuming $I \subseteq O \wedge RL \wedge \text{gflow}$: For $I = O$, Theorem~\ref{thm:io_g_to_rl_yz} implies condition $\lambda \equiv YZ$. For $I \subsetneq O$ this condition is not implied, see Fig.~\ref{fig:counter_gflow}(g).
\end{itemize}

As mentioned in Sec.~\ref{sec:yz_mbqc_with_gflow} of the main text, Theorem~\ref{thm:mod_isaac_thm_1} makes two modifications to the original version, Theorem~2 in Ref.~\cite{smith_parity_2024}. There the claim is that a simple connected labeled open graph $(G,I,O,\lambda)$ with $|I| = |O|$ and using only $YZ$-plane measurements has gflow if and only if $G$ is bipartite with $I$ forming one part. The counterexample in Fig.~\ref{fig:counter_gflow}(h) shows that the graph need not be bipartite with respect to $I$ if the assumptions $\vert I \vert = \vert O \vert$ and existence of gflow are fulfilled. For the other implication, assuming only $\vert I \vert = \vert O \vert$ is too weak as shown in the counterexample in Fig.~\ref{fig:counter_gflow}(i). 

These observations lead to the modified claim in Theorem~\ref{thm:mod_isaac_thm_1} where we consider a more general class of graphs, the RL graphs, and instead of $|I| = |O|$ assume $I = O$ right away. For the implication from left to right the assumption $I = O$ can indeed be relaxed to $\vert I \vert = \vert O \vert$ because of Lemma~\ref{thm:inputs_gflow}. This leads to another way of modifying the original theorem \footnote{Private communication with Isaac Smith, one of the authors of Ref.~\cite{smith_parity_2024}.}: Consider a labeled open graph without fixed outputs $(G,I,\lambda\equiv YZ)$ and with no edges between the inputs. Then \emph{there exists a choice of outputs} $O$ with $\vert I \vert = \vert O \vert$ such that there exists a gflow if and only if $G$ is bipartite with respect to $I$.

\section{Explicit LHZ triangle correspondence}
\label{sec:cluster_to_lhz}

The same correspondence as discussed in Sec.~\ref{sec:parity_architecture} can also be shown more explicitly by invoking the rule for $X$ basis measurements \cite{hein_entanglement_2006,hein_multi-party_2004}. For the measurement of qubit $a$ one also needs to choose one of its neighbors $b_0$ to describe the resulting graph state. For a given qubit $a$ in the triangular cluster segment, one can always choose its ``north'' neighbor for $b_0$. The explicit sequence of steps for the LHZ triangle with five base qubits is demonstrated in Fig.~\ref{fig:cluster_to_lhz_5}. Note that the $X$ basis measurement also entails an extra Hadamard and $Z$ operators applied to some qubits. These operators can be merged with the $XY$-plane measurements that are to be performed as a part of the measurement pattern of the LHZ computation. It is precisely these extra Hadamards that turn $XY$-plane measurements into the $YZ$-plane measurements, reproducing thus the results of Ref.~\cite{smith_parity_2024}.

\begin{figure*}
    \centering
    \includegraphics[width=\linewidth]{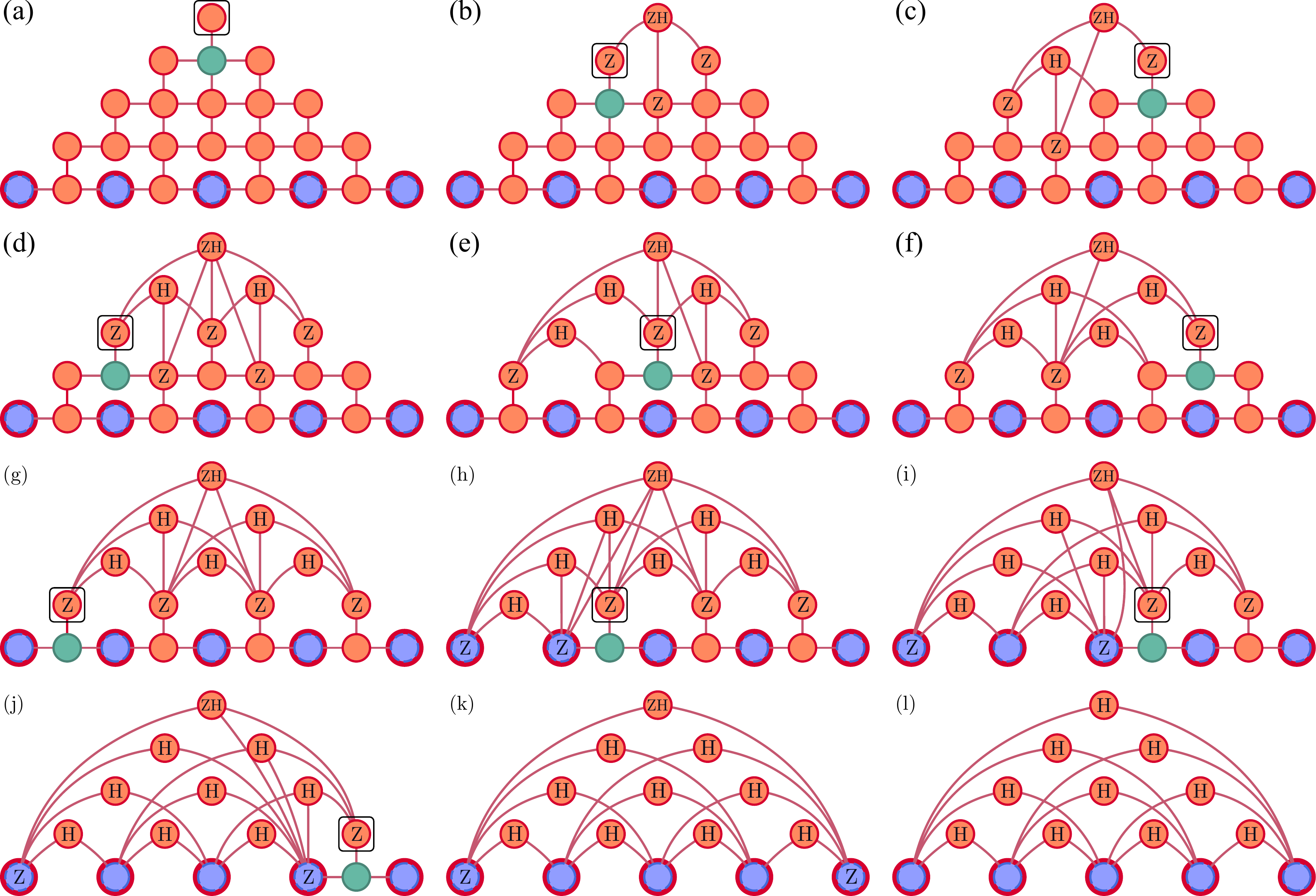}
    \caption{Series of steps that turn a triangular patch of a cluster state into the LHZ triangle. The overall modification consists in measuring all qubits in one sublattice in the $X$ basis. These measurements are taken to be $\ket{+}$ projections for convenience. The $\ket{-}$ projections would correspond to extra Pauli gates applied to some qubits that can be merged into the subsequent measurements. Each step except the last one is the $X$ basis measurement of the qubit highlighted in green in the previous step, while the qubit in the square is the corresponding $b_0$ qubit, see the text. The last step is the application of the stabilizer of the top-most qubit. The operators written in some nodes represent the extra gates applied to the given qubits on top of the $\cz$ gates.}
    \label{fig:cluster_to_lhz_5}
\end{figure*}

\section{$YZ$-plane patterns with Pauli flow}
\label{sec:pauli_flow_pushed}

Note that the inputs of any measurement pattern that has Pauli flow and that are not simultaneously outputs must be measured in the $XY$-plane -- either in the Pauli or rotated basis \cite{mitosek_algebraic_2024} (for gflow this is proven in Refs.~\cite{hamrit_reversibility_2015, backens_there_2021}). For patterns where all non-outputs are measured in the $YZ$-plane we therefore obtain the following statement. Recall that $\lambda \subseteq YZ$ refers to the case where all non-outputs are measured in the rotated $YZ$-plane basis, $Y$ basis, or $Z$ basis. 

\begin{lem}[Restatement of Lemma~\ref{thm:inputs_pauli_flow}]
    Any labeled open graph $(G, I, O, \lambda)$ with Pauli flow and $\lambda \subseteq YZ$ must be such that each input is either measured in the $Y$ basis or belongs simultaneously to the outputs.
\end{lem}
\begin{proof}
    For completeness we prove that all inputs must be measured in the $XY$-plane, the actual claim is a trivial consequence. Since the correction function $p$ maps to the non-inputs, for any input $u$ it holds that $u \not\in p(u)$ and so $u$ cannot be measured in any rotated basis in the $XZ$ and $YZ$ planes nor in the $Z$ basis. This follows from conditions~\ref{itm:pauli_flow_xy} through \ref{itm:pauli_flow_z} of the definition of Pauli flow.
\end{proof}

\begin{figure}
    \centering
    \includegraphics[width=.7\linewidth]{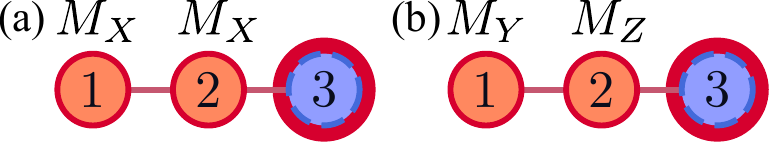}
    \caption{Counterexamples with Pauli flow. The graphics follows the convention of Fig.~\ref{fig:counter_gflow}. %
    (a) Graph with $I = O$ and Pauli flow that is neither of the RL form nor are non-outputs measured in the $YZ$-plane. Pauli flow reads $p(1) = \{2\}$, $p(2) = \{1\}$ with $1 \prec 3$.
    (b) Labeled open graph with $I = O$ and Pauli flow, where all non-outputs are measured in the $YZ$-plane and yet the graph is not of the RL form. Pauli flow reads $p(1) = \{2\}$, $p(2) = \{1,2\}$ with $1 \prec 2 \prec 3$. %
    }
    \label{fig:counter_pauli}
\end{figure}

As discussed in Sec.~\ref{sec:yz_mbqc_with_gflow}, the uniform determinism of a measurement pattern, captured by the notion of gflow, puts very stringent constraints on the form of the underlying open graph. When instead of gflow a more general notion of Pauli flow is considered and $\lambda \equiv YZ$ is replaced by $\lambda \subseteq YZ$, the three theorems derived in Sec.~\ref{sec:yz_mbqc_with_gflow} for gflow must be reanalyzed. Theorem~\ref{thm:io_g_to_rl_yz} does not hold anymore, cf.~Fig.~\ref{fig:counter_pauli}(a), while Theorem~\ref{thm:rl_yz_impl} remains valid since any gflow is a special case of a Pauli flow. The implication from right to left in Theorem~\ref{thm:mod_isaac_thm_1} still holds for the same reason, while the converse implication holds no more, refer to the counterexample in Fig.~\ref{fig:counter_pauli}(b).

For completeness, in the rest of this section the correction strategy of the pattern in Sec.~\ref{sec:univ_yz_pattern} is presented, where all the Pauli measurements are made in the very first time step and the rest of computation consists of the rotated measurements of qubits $q_2$ in each unit cell. Since no confusion with parity labels can arise here, instead of $q_j$ we refer to a given qubit merely by its numerical index $j$ from now on. 
The sign of the rotation angle in general depends on previous measurement results. A graph state is defined by the elementary stabilizers $K_j$ for each non-input qubit $j$, see Sec.~\ref{sec:prelims}. The correction strategy in general consists in the application of the correction stabilizer $Q_j$ to qubit $j$ to counteract the undesirable measurement outcome. The correction stabilizers for each non-output qubit of the unit cell read as follows:
\begin{itemize}
    \item $Q_1$ applies $Z$ to qubit 1 and reads $Q_1 = K_3 Q_3$.
    \item $Q_2$ applies $X$ to qubit 2 and reads $Q_2 = K_2 Q_3$.
    \item $Q_3$ depends on the exact choice of Pauli measurements in the unit cell and is studied in detail below.
    \item $Q_4$ applies $X$ to qubit 4 and reads $Q_4 = K_4$ irrespective of whether 4 is measured in $Y$ or $Z$.
    \item $Q_5$ applies $X$ to qubit 5 and reads $Q_5 = K_5$ irrespective of whether 5 is measured in $Y$ or $Z$.
\end{itemize}
As for the correction stabilizer of qubit 3, it applies $Z$ to 3 and so includes $K_6$. This elementary stabilizer not only affects the other measured qubits of the cell but also has neighbors in other cells. A more detailed discussion is therefore necessary and is summarized in the table below, where the first three columns denote the measurement basis of qubits 4, 5, and $5'$, respectively, where the primed numbers denote the corresponding qubits of the unit cell that lies to the north of the cell in question. If the latter is already on the border, no $5'$ exists and the expressions below remain the same except that $5'$ is excluded. The last two columns then show the corresponding correction stabilizer $Q_3$:

\begin{center}
    \begin{tabular}{c|c|c|l|l}
     $M^{(4)}$ & $M^{(5)}$ & $M^{(5')}$ & $Q_3$ stabilizer & explicitly (phase ignored) \\ \hline
     Y & Y & Y & $K_4 K_5 K_{5'} K_6$ & $Y_4 Y_5 Y_{5'} Y_6 Z_3 Z_{6'} Z_7$ \\
     Y & Y & Z & $K_4 K_5 K_6$ & $X_6 Y_4 Y_5 Z_3 Z_{5'} Z_{7}$ \\
     Y & Z & Y & $K_4 K_{5'} K_6$ & $X_6 Y_4 Y_{5'} Z_3 Z_5 Z_{6'}$ \\
     Y & Z & Z & $K_4 K_6$ & $Y_4 Y_6 Z_3 Z_5 Z_{5'}$ \\
     Z & Z & Y & $K_{5'} K_6$ & $Y_{5'} Y_6 Z_3 Z_4 Z_5 Z_{6'}$ \\
     Z & Z & Z & $K_6$ & $X_6 Z_3 Z_4 Z_5 Z_{5'}$
\end{tabular}
\end{center}

Note that combinations of measurements $(Z, Y, Z)$ and $(Z, Y, Y)$ are not used for the implementation of universal gates and so do not have to be considered for the correction strategy. The stabilizers presented above are such that their individual Pauli terms either lie on the output qubits 6 and 7, or the term coincides with the Pauli that is being measured (except of course for the qubit that is being corrected). As a result, such a stabilizer does not affect the measurement results. The only exception is the stabilizer for qubit 1, whose term that lies on qubit 2 is $Z$ and so flips the measurement angle of that qubit. So far we discussed only a single unit cell with its north neighbor. For this entire correction strategy to work also for the entire sheet of cells, one has to show that the Paulis that accumulate on the outputs of one cell propagate easily through the subsequent cells. The correction stabilizers for such a sheet can be constructed iteratively from the outputs. The last column of cells has the correction strategy exactly as explained above. For the second-to-last column, there might be any combination of the $X$, $Y$, or $Z$ Pauli operators on the outputs of the first column of cells. Moreover, each output qubit from the first column also has an extra neighboring qubit to the east, which is not present for a single unit cell considered above and which represents qubit 3 of the subsequent cell. This extra neighbor gets $Z$ correction whenever the associated output gets either $X$ or $Y$ correction. One can thus cancel it by applying the correction stabilizer for qubit 3 of the subsequent cell. The $Z$ correction on the output can be removed by the correction stabilizer of the output when understood as the input of the next cell. The $Y$ correction gets merged with the $Y$ measurement so no adjustment is necessary. Finally, if the $X$ correction is to be applied, one applies the correction stabilizer for the output as well to turn the $X$ correction into the $Y$ one.

\section{Universal $XZ$-plane pattern}
\label{sec:explicit_xz_univ}

\begin{figure}
    \centering
    \includegraphics[width=\linewidth]{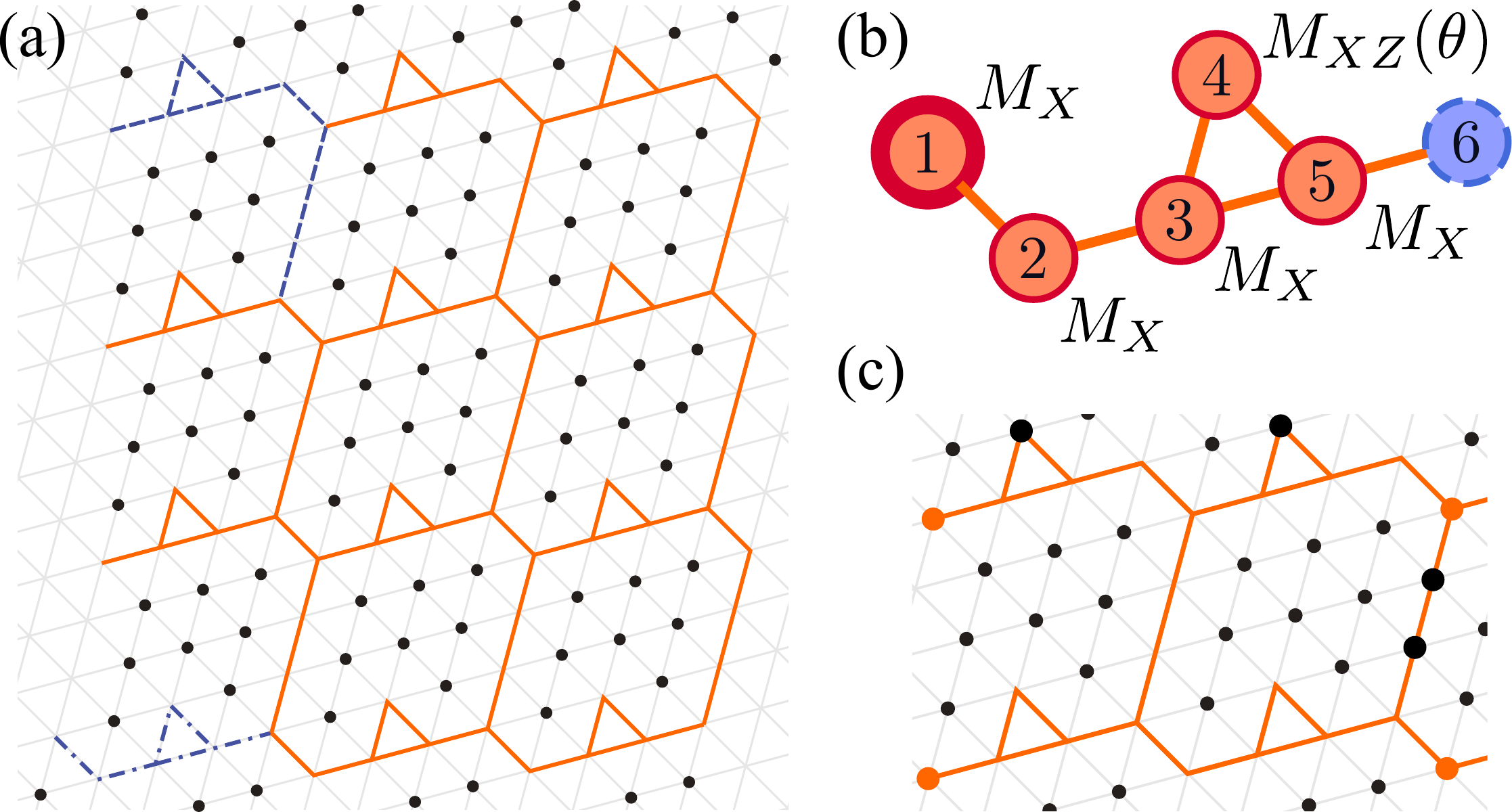}
    \caption{Universal real-valued MBQC in the $XZ$ plane on a triangular grid. Qubits marked black are measured in the $Z$ basis. (a) The resulting graph state that serves as a universal resource can be understood as the tessellation of two elementary subgraphs, highlighted by dashed and dot-dashed lines, respectively. (b) Second ``thorn''-shaped subgraph with measurement bases shown and qubits indexed for convenience, where 1 is the input and 6 the unmeasured output. (c) The patch of the graph state that implements $\cz$. The two left-most and two right-most orange (gray) dots denote the inputs and outputs, respectively. The vertical edges adjacent to the very top and bottom of the pattern shown are removed by measuring the qubits above or below in the $Z$ basis.}
    \label{fig:trig_grid_univ}
\end{figure}

The universal real-valued quantum computation can be driven by single-qubit measurements in the $XZ$ plane made on qubits that form a regular triangular grid \cite{mhalla_graph_2012}. Here we present a somewhat streamlined version of the original argument. The universality is to be understood such that the measurements implement a set of gates that are \emph{computationally universal} \cite{aharonov_simple_2003}. Such gates might not be \emph{strictly universal}, which means that the subgroup generated by them is not necessarily dense in $\mathrm{SU}(2^n)$ for $n$ qubits. These gates can nevertheless efficiently simulate a strictly universal set of gates \cite{aharonov_simple_2003}.

Consider a regular triangular grid, where many qubits are measured in the $Z$ basis to arrive at the sparse graph state depicted with colored lines in Fig.~\ref{fig:trig_grid_univ}(a). The underlying graph of this state can be seen as being tiled by two kinds of elementary subgraphs that are highlighted with dashed and dash-dotted lines. The latter subgraph resembles the rotated ``thorn'' rune, see Fig.~\ref{fig:trig_grid_univ}(b), and implements different single-qubit gates depending on the measurement angle $\theta$ of the apex qubit with index 4. To see how, one can first apply the local-complementation rule for graph states to the apex qubit \cite{hein_entanglement_2006,mhalla_graph_2012} that effectively removes the edge between qubits 3 and 5 and turns this way the thorn-shaped state into a linear cluster, where 1 and 2 are measured in the $X$ basis, qubits 3 and 5 in the $Y$ basis, and qubit 4 in the $XY$-plane. The effect of such measurements on this linear cluster is easy to compute with the conclusion that for any real $\theta$, the gate applied to the input qubit reads $R_Y (\theta - \pi/2) H$, which up to some Pauli corrections reduces to $\ident$ and $H$ for $\theta = 0$ and $\theta = \pi/2$, respectively. The subgraph highlighted in dashed lines can be seen as being made out of the thorn and the vertical part that consists of four qubits connected by three edges. When the two qubits in the middle are measured in the $Z$ basis, the overall operation of this vertical segment is the identity. When they are measured in the $X$ basis, the effective action is the $\cz$ gate. One has to be a little careful in the case of $\cz$, however, as the corresponding linear cluster is not attached to the ends of the two neighboring thorns. For this reason, it is easier to consume not two but four cells of the graph state to implement a single $\cz$ as demonstrated in Fig.~\ref{fig:trig_grid_univ}(c). Almost all qubits in the associated graph state are measured in the $X$ basis. Note that if also the apexes of the two bottom thorns are measured in the $Z$ basis, the resulting operation is CNOT. As noted in Ref.~\cite{mhalla_graph_2012}, $\{ H, P(\alpha), \cz \}$ with $\alpha \in [0, 2 \pi)$ represent a universal set of gates for real-valued quantum computation, where $P(\alpha) = X R_Y(\alpha)$. Since $X = H R_Y(\pi/2)$, the universal set can be chosen as $\{ H, R_Y(\alpha), \cz \}$. The pattern discussed above can implement all these gates, thus proving its universality for real-valued quantum computing and hence also for the complex-valued one \cite{mhalla_graph_2012,aharonov_simple_2003}.

\end{document}